\documentclass{lmcs}
\pdfoutput=1

\usepackage[utf8]{inputenc}

\usepackage{lastpage}
\lmcsdoi{16}{4}{15}
\lmcsheading{}{\pageref{LastPage}}{}{}%
{Oct.~09,~2019}{Dec.~15,~2020}{}

\pdfoutput=1


\usepackage[british]{babel}
\usepackage{graphicx}
\usepackage{lscape}
\usepackage{amsmath,amssymb}
\usepackage{scalerel,rotating}
\usepackage{textgreek,textcomp,stmaryrd} 
\usepackage{centernot}
\usepackage{xspace}
\xspaceaddexceptions{]\}}
\usepackage[all]{foreign}

\newcommand{\Ami}{Ami\xspace}
\newcommand{\Boe}{Bo\'{e}\xspace}
\newcommand{\Cat}{Cat\xspace}

\newcommand{\emoji}[2][1em]{\ensuremath{\vcenter{%
\hbox{\includegraphics[width=#1]{#2}}}}\xspace}
\newcommand{\twemoji}[2][1em]{\emoji[#1]{twemoji/2/assets/#2.ai}}
\newcommand{\ami}[1][1em]{\twemoji[#1]{1f9d1-1f3fd}}
\newcommand{\boe}[1][1em]{\twemoji[#1]{1f469-1f3fd}}
\newcommand{\cat}[1][1em]{\twemoji[#1]{1f467-1f3fd}}
\newcommand{\sliceofcake}[1][1em]{\twemoji[#1]{1f370}}
\newcommand{\doughnut}[1][1em]{\twemoji[#1]{1f369}}
\newcommand{\cake}[1][1em]{\twemoji[#1]{1f382}}
\newcommand{\nope}[1][1em]{\twemoji[#1]{1f64c-1f3fd}}
\newcommand{\money}[1][1em]{\twemoji[#1]{1f4b0}}
\newcommand{\bill}[1][1em]{\twemoji[#1]{1f4b7}}

\newcommand{\store}[1][1em]{\twemoji[#1]{1f3ea}}

\definecolor{tmcolor}{HTML}{C02F1D}
\definecolor{tycolor}{HTML}{1496BB}
\providecommand{\tm}[1]{\textcolor{tmcolor}{\ensuremath{\normalfont#1}}}
\providecommand{\ty}[1]{\textcolor{tycolor}{\ensuremath{\normalfont#1}}}

\usepackage{bussproofs}
\EnableBpAbbreviations

\providecommand{\seq}[2][]{\ensuremath{\tm{#1}\;\vdash\;\ty{#2}}}
\providecommand{\tmty}[2]{\ensuremath{\tm{#1}\colon\!\ty{#2}}}
\providecommand{\NOM}[1]{\RightLabel{\textsc{#1}}}
\providecommand{\SYM}[1]{\RightLabel{\ensuremath{#1}}}
\newenvironment{prooftree*}{\leavevmode\hbox\bgroup}{\DisplayProof\egroup}

\providecommand{\reducesto}[3][]{\ensuremath{\tm{#2}\overset{#1}{\Longrightarrow}\tm{#3}}}

\providecommand{\dhcp}{\ensuremath{\text{DHCP}}\xspace}
\providecommand{\hcp}{\ensuremath{\text{HCP}}\xspace}
\providecommand{\nodcap}{\ensuremath{\hcp_{\text{ND}}}\xspace}
\providecommand{\piDILL}{\textpi DILL\xspace}
\providecommand{\SILL}{$\text{SILL}$\xspace}
\providecommand{\SILLS}{$\text{\SILL}_{S}$\xspace}
\providecommand{\cp}{CP\xspace}

\providecommand{\ppar}{\ensuremath{\parallel}}

\providecommand{\piSend}[3]{\ensuremath{#1[ #2 ].#3}}
\providecommand{\piUSend}[3]{\ensuremath{#1\langle #2 \rangle.#3}}
\providecommand{\piRecv}[3]{\ensuremath{#1( #2 ).#3}}
\providecommand{\piPar}[2]{\ensuremath{#1 \ppar #2}}
\providecommand{\piNew}[3]{\ensuremath{(\nu #1#2)#3}}

\providecommand{\piHalt}[0]{\ensuremath{0}}
\providecommand{\piSub}[3]{\ensuremath{#3\{#1/#2\}}}

\providecommand{\cpLink}[2]{\ensuremath{#1{\leftrightarrow}#2}}

\providecommand{\cpSend}[4]{\ensuremath{#1[#2].(\piPar{#3}{#4})}}
\providecommand{\cpRecv}[3]{\ensuremath{#1(#2).#3}}
\providecommand{\cpWait}[2]{\ensuremath{#1().#2}}
\providecommand{\cpHalt}[1]{\ensuremath{#1[].0}}
\providecommand{\cpInl}[2]{\ensuremath{#1\triangleleft\texttt{inl}.#2}}
\providecommand{\cpInr}[2]{\ensuremath{#1\triangleleft\texttt{inr}.#2}}
\providecommand{\cpCase}[3]{\ensuremath{#1\triangleright\{\texttt{inl}:#2;\texttt{inr}:#3\}}}
\providecommand{\cpAbsurd}[1]{\ensuremath{#1\triangleright\{\}}}
\providecommand{\cpSub}[3]{\ensuremath{\piSub{#1}{#2}{#3}}}

\newcommand*{\queue}[1]{\ensuremath{\text{queue} \; #1}}%
\newcommand*{\acquire}[1]{\ensuremath{{\uparrow^{S}_{L}}#1}}%
\newcommand*{\release}[1]{\ensuremath{{\downarrow^{S}_{L}}#1}}%
\providecommand{\parr}{\mathbin{\bindnasrepma}}
\providecommand{\with}{\mathbin{\binampersand}}
\providecommand{\plus}{\ensuremath{\oplus}}
\providecommand{\tens}{\ensuremath{\otimes}}
\providecommand{\one}{\ensuremath{\mathbf{1}}}
\providecommand{\nil}{\ensuremath{\mathbf{0}}}

\providecommand{\hsep}{\ensuremath{\parallel}}
\providecommand{\emptyhypercontext}{\varnothing}

\providecommand{\hcpInfAx}{%
  \begin{prooftree*}
    \AXC{$\vphantom{\seq[ Q ]{ \Delta, \tmty{y}{A^\bot} }}$}
    \NOM{Ax}
    \UIC{$\seq[ \cpLink{x}{y} ]{ \tmty{x}{A}, \tmty{y}{A^\bot} }$}
  \end{prooftree*}}
\providecommand{\hcpInfCut}{%
  \begin{prooftree*}
    \AXC{$\seq[{ P }]{
        \mathcal{G}\hsep
        \Gamma, \tmty{x}{A} \hsep
        \Delta, \tmty{x'}{A^\bot}
      }$}
    \NOM{Cut}
    \UIC{$\seq[{ \piNew{x}{x'}{P} }]{
        \mathcal{G}\hsep
        \ty{\Gamma}, \ty{\Delta}
      }$}
  \end{prooftree*}}
\providecommand{\hcpInfMix}{%
  \begin{prooftree*}
    \AXC{$\seq[ P ]{\mathcal{G} }$}
    \AXC{$\seq[ Q ]{\mathcal{H} }$}
    \NOM{H-Mix}
    \BIC{$\seq[ \piPar{P}{Q} ]{
        \mathcal{G} \hsep \mathcal{H} }$}
  \end{prooftree*}}
\providecommand{\hcpInfHalt}{%
  \begin{prooftree*}
    \AXC{$\vphantom{\seq[ Q ]{ \Delta, \tmty{y}{A^\bot} }}$}
    \NOM{H-Mix$_0$}
    \UIC{$\seq[{ \piHalt }]{ \emptyhypercontext }$}
  \end{prooftree*}}

\providecommand{\hcpInfBoundTens}{%
  \begin{prooftree*}
    \AXC{$\seq[{ P }]{
        \ty{\Gamma}, \tmty{y}{A} \hsep \ty{\Delta}, \tmty{x}{B}
      }$}
    \SYM{\tens}
    \UIC{$\seq[{ \piSend{x}{y}{P} }]{
        \ty{\Gamma}, \ty{\Delta}, \tmty{x}{A \tens B}
      }$}
  \end{prooftree*}}
\providecommand{\hcpInfUnboundTens}{%
  \begin{prooftree*}
    \AXC{$\seq[{ P }]{
        \ty{\Gamma}, \tmty{x}{B}
      }$}
    \UIC{$\seq[{ \piUSend{x}{y}{P} }]{
        \ty{\Gamma}, \tmty{x}{A \tens B}, \tmty{y}{A^\bot}
      }$}
  \end{prooftree*}}
\providecommand{\hcpInfParr}{%
  \begin{prooftree*}
    \AXC{$\seq[ P ]{%
        \Gamma , \tmty{y}{A} , \tmty{x}{B} }$}
    \SYM{(\parr)}
    \UIC{$\seq[ \cpRecv{x}{y}{P} ]{
        \Gamma , \tmty{x}{A \parr B} }$}
  \end{prooftree*}}
\providecommand{\hcpInfOne}{%
  \begin{prooftree*}
    \AXC{$\seq[{ P }]{\emptyhypercontext}$}
    \SYM{\one}
    \UIC{$\seq[{ \piSend{x}{}{P} }]{
        \tmty{x}{\one}
      }$}
  \end{prooftree*}}
\providecommand{\hcpInfBot}{%
  \begin{prooftree*}
    \AXC{$\seq[ P ]{
        \Gamma }$}
    \SYM{(\bot)}
    \UIC{$\seq[ \cpWait{x}{P} ]{
        \Gamma , \tmty{x}{\bot} }$}
  \end{prooftree*}}
\providecommand{\hcpInfPlus}[1]{%
  \ifdim#1pt=1pt
  \begin{prooftree*}
    \AXC{$\seq[ P ]{
        \Gamma , \tmty{x}{A} }$}
    \SYM{(\plus_1)}
    \UIC{$\seq[{ \cpInl{x}{P} }]{
        \Gamma , \tmty{x}{A \plus B} }$}
  \end{prooftree*}
  \else%
  \ifdim#1pt=2pt
  \begin{prooftree*}
    \AXC{$\seq[ P ]{
        \Gamma , \tmty{x}{B} }$}
    \SYM{(\plus_2)}
    \UIC{$\seq[ \cpInr{x}{P} ]{
        \Gamma , \tmty{x}{A \plus B} }$}
  \end{prooftree*}
  \else%
  \fi%
  \fi%
}
\providecommand{\hcpInfWith}{%
  \begin{prooftree*}
    \AXC{$\seq[ P ]{
        \Gamma , \tmty{x}{A} }$}
    \AXC{$\seq[ Q ]{
        \Gamma, \tmty{x}{B} }$}
    \SYM{(\with)}
    \BIC{$\seq[ \cpCase{x}{P}{Q} ]{
        \Gamma , \tmty{x}{A \with B} }$}
  \end{prooftree*}}
\providecommand{\hcpInfTop}{%
  \begin{prooftree*}
    \AXC{}
    \SYM{(\top)}
    \UIC{$\seq[ \cpAbsurd{x} ]{ \Gamma, \tmty{x}{\top} }$}
  \end{prooftree*}}
\providecommand{\hcpInfNil}{%
  (no rule for \ty{\nil})}

\providecommand{\ncSrv}[3]{\ensuremath{{\star}{#1}(#2).#3}}
\providecommand{\ncCnt}[3]{\ensuremath{{\star}{#1}[#2].#3}}
\providecommand{\ncPool}[2]{\ensuremath{(\piPar{#1}{#2})}}

\providecommand{\give}[2][]{\ensuremath{{ ? }_{#1}{#2}}}
\providecommand{\take}[2][]{\ensuremath{{ ! }_{#1}{#2}}}

\providecommand{\ncRedBetaStar}[1]{\textsc{E-Request}\xspace}

\providecommand{\ncInfTake}[1]{%
  \begin{prooftree*}
    \AXC{$\seq[{ P }]{
        \Gamma, \tmty{y}{A} }$}
    \SYM{(\take[1]{})}
    \UIC{$\seq[{ \ncCnt{x}{y}{P} }]{
        \Gamma, \tmty{x}{\take[1]{A}} }$}
  \end{prooftree*}}
\providecommand{\ncInfGive}[1]{%
  \begin{prooftree*}
    \AXC{$\seq[{ P }]{
        \Gamma, \tmty{y}{A} }$}
    \SYM{(\give[1]{})}
    \UIC{$\seq[{ \ncSrv{x}{y}{P} }]{
        \Gamma, \tmty{x}{\give[1]{A}} }$}
  \end{prooftree*}}
\providecommand{\ncInfCont}{%
  \begin{prooftree*}
    \AXC{$\seq[{ P }]{
        \mathcal{G} \hsep
        \Gamma, \tmty{x}{\give[m]{A}} , \tmty{x'}{\give[n]{A}} }$}
    \SYM{\textsc{Cont}_{?}}
    \UIC{$\seq[{ \cpSub{x}{x'}{P} }]{
        \mathcal{G} \hsep
        \Gamma, \tmty{x}{\give[m+n]{A}} }$}
  \end{prooftree*}}
\providecommand{\ncInfPool}{%
  \begin{prooftree*}
    \AXC{$\seq[{ P }]{
        \mathcal{G} \hsep
        \Gamma, \tmty{x}{\take[m]{A}} \hsep
        \Delta, \tmty{x'}{\take[n]{A}} }$}
    \SYM{\textsc{Cont}_{!}}
    \UIC{$\seq[{ \cpSub{x}{x'}{P} }]{
        \mathcal{G} \hsep
        \Gamma, \Delta, \tmty{x}{\take[m+n]{A}} }$}
  \end{prooftree*}}

\providecommand{\hcpEquivLinkComm}{\textsc{SC-LinkSwap}\xspace}
\providecommand{\hcpEquivMixAss}{\textsc{SC-ParAssoc}\xspace}
\providecommand{\hcpEquivMixComm}{\textsc{SC-ParComm}\xspace}
\providecommand{\hcpEquivMixHalt}{\textsc{SC-ParNil}\xspace}
\providecommand{\hcpEquivScopeExt}{\textsc{SC-ResExt}\xspace}
\providecommand{\hcpEquivNewHalt}{\textsc{SC-ResNil}\xspace}
\providecommand{\hcpEquivNewComm}{\textsc{SC-ResComm}\xspace}

\providecommand{\hcpRedAxCut}[1]{\textsc{E-Link}\xspace}%
\providecommand{\hcpRedBetaTensParr}{\textsc{E-Send}\xspace}%
\providecommand{\hcpRedBetaOneBot}{\textsc{E-Close}\xspace}%
\providecommand{\hcpRedBetaPlusWith}[1]{\textsc{E-Sel}\textsubscript{#1}\xspace}%
\providecommand{\hcpRedGammaNew}{\textsc{E-LiftRes}\xspace}%
\providecommand{\hcpRedGammaMix}{\textsc{E-LiftPar}\xspace}%
\providecommand{\hcpRedGammaEquiv}{\textsc{E-LiftSC}\xspace}%

\usepackage{hyperref}
\usepackage{cleveref}
\crefname{thm}{theorem}{theorems}
\Crefname{thm}{Theorem}{Theorems}
\crefname{lem}{lemma}{lemmas}
\Crefname{lem}{Lemma}{Lemmas}
\crefname{defi}{definition}{definitions}
\Crefname{defi}{Definition}{Definitions}
\newtheorem*{remark}{Remark}
\usepackage{doi}
\usepackage[strings]{underscore}
\usepackage{upgreek}

\begin{document}

\title{Towards Races in Linear Logic}

\author[W.~Kokke]{Wen Kokke\rsuper{a}}
\address{\lsuper{a}University of Edinburgh, Edinburgh, UK}
\email{\{wen.kokke,wadler\}@ed.ac.uk}

\author[G.~Morris]{J.\ Garrett Morris\rsuper{b}}
\address{\lsuper{b}University of Kansas, Lawrence, KS, USA}
\email{garrett@ittc.ku.edu}

\author[P.~Wadler]{Philip Wadler\rsuper{a}}

\begin{abstract}
  Process calculi based in logic, such as \piDILL and CP, provide a foundation for deadlock-free concurrent programming, but exclude non-determinism and races. \hcp is a reformulation of CP which addresses a fundamental shortcoming: the fundamental operator for parallel composition from the \textpi-calculus does not correspond to any rule of linear logic, and therefore not to any term construct in CP.

  We introduce \nodcap, which extends \hcp with a novel account of non-determinism. Our approach draws on bounded linear logic to provide a strongly-typed account of standard process calculus expressions of non-determinism. We show that our extension is expressive enough to capture many uses of non-determinism in untyped calculi, such as non-deterministic choice, while preserving \hcp's meta-theoretic properties, including deadlock freedom.  
\end{abstract}

\maketitle

\section{Introduction}\label{sec:introduction}

Consider the following scenario:
\begin{quote}
  \Ami and \Boe are working from home one morning when they each get a craving for a slice of cake. Being denizens of the web, they quickly find the nearest store which does home deliveries. Unfortunately for them, they both order their cake at the \emph{same} store, which has only one slice left. After that, all it can deliver is disappointment.
\end{quote}
This is an example of a \emph{race condition}. We can model this scenario in the \textpi-calculus, where \ami, \boe and \store are processes modelling \Ami, \Boe and the store, and \sliceofcake and \nope are channels giving access to a slice of cake and disappointment, respectively. This process has two possible outcomes: either \Ami gets the cake, and \Boe gets disappointment, or vice versa. 
\begin{center}
  \(
  \begin{array}{c}
    \tm{(\piPar%
    {\piPar{\piRecv{x}{y}{\ami}}{\piRecv{x}{z}{\boe}}}
    {\piUSend{x}{\sliceofcake}{\piUSend{x}{\nope}{\store}}}
    )}
    \\[1ex]
    \rotatebox[origin=c]{270}{$\Longrightarrow^{\star}$}
    \\[1ex]
    \tm{(\piPar%
    {\piPar
    {\piSub{\sliceofcake}{y}{\ami}}
    {\piSub{\nope}{z}{\boe}}
    }
    {\store}
    )}
    \quad
    \text{or}
    \quad
    \tm{(\piPar{\store}{\piPar{\piSub{\nope}{y}{\ami}}{%
    \piSub{\sliceofcake}{z}{\boe}}})}
  \end{array}
  \)
\end{center}
While \Ami or \Boe may not like all of the outcomes, it is the store which is responsible for implementing the online delivery service, and the store is happy with either outcome. Thus, the above is an interaction we would like to be able to model.

Now consider another scenario, which takes place \emph{after} \Ami has already bought the cake:
\begin{quote}
  \Boe is \emph{really} disappointed when she finds out the cake has sold out. \Ami, always looking to make some money, offers to sell the slice to her for a profit. \Boe agrees to engage in a little bit of back-alley cake resale, but sadly there is no trust between the two. \Ami demands payment first. \Boe would rather get her slice of cake before she gives \Ami the money.
\end{quote}
This is an example of a \emph{deadlock}. We can also model this scenario in the \textpi-calculus, where \bill\ is a channel giving access to some adequate amount of money. 
\begin{center}
  \(
  \begin{array}{c}
    \tm{(\piPar{%
    \piRecv{x}{z}{\piUSend{y}{\sliceofcake}{\ami}}
    }{%
    \piRecv{y}{w}{\piUSend{x}{\bill}{\boe}}
    })}
    \quad
    \centernot\Longrightarrow^{\star}
  \end{array}  
  \)
\end{center}
The above process does not reduce. As both \Ami and \Boe would prefer the exchange to be made, this interaction is desired by \emph{neither}. Thus, the above is an interaction we would like to exclude.

Session types~\cite{honda1993} statically guarantee that concurrent programs, such as those above, respect communication protocols. Session-typed calculi with logical foundations, such as \piDILL~\cite{caires2010} and CP~\cite{wadler2012}, obtain deadlock freedom as a result of a close correspondence with logic. These systems, however, also rule out non-determinism and race conditions. In this paper, we demonstrate that logic-inspired type systems need not rule out races.

We present \nodcap, an extension of \hcp with a novel account of non-determinism and races. Inspired by bounded linear logic~\cite{girard1992}, we introduce a form of shared channels in which the type of a shared channel tracks how many times it is reused. As in the untyped \textpi-calculus, sharing introduces the potential for non-determinism. We show that our approach is sufficient to capture practical examples of races, such as an online store, as well as other formal characterizations of non-determinism, such as non-deterministic choice.  However, \nodcap does not lose the meta-theoretical benefits of \hcp: we show that it enjoys termination and deadlock-freedom.

An important limitation of our work is that types in \nodcap explicitly count the potential races on a channel.  It works fine when there are two or three races, but not $n$ for an arbitary $n$.  The latter case is obviously important, and we see the main value of our work as a stepping stone to this more general case.

\nodcap is based on \hcp~\cite{kokke2018tlla,kokke2019pacmpl}. \hcp is a reformulation of CP which addresses a fundamental shortcoming: the fundamental operator for parallel composition from the \textpi-calculus does not correspond to any rule of linear logic, and therefore not to any term construct in CP.

This paper proceeds as follows. In \cref{sec:local-choice}, we discuss recent approaches to non-determinism in logic-inspired session-typed process calculi. In \cref{sec:cp-revisited}, we introduce a variant of \cp and prove progress and preservation. In \cref{sec:cpnd}, we introduce \nodcap.  In \cref{sec:leftovers}, we discuss cuts with leftovers. In \cref{sec:manifest}, we discuss the relation to manifest sharing~\cite{balzer2017}. Finally, in \cref{sec:conclusion}, we conclude with a discussion of the work done in this paper and potential avenues for future work.

\begin{remark}[Variants of \hcp]
  There are two variants of \hcp: a~version with delayed actions, introduced by Kokke, Montesi, and Peressotti~\cite{kokke2019pacmpl}, and a version without delayed actions, introduced by Kokke, Montesi, and Peressotti~\cite{kokke2018tlla}.
  Delayed actions are not an essential part of HCP, but significantly complicate the theory. Therefore, we base our work on the variant without delayed actions. For typographical simplicity, we will refer to the system without delayed actions as \hcp, instead of $\hcp^{-}$. Should we need to refer to the system with delayed actions, we will use \dhcp.
\end{remark}

\section{Non-determinism, Logic, and Session Types}\label{sec:local-choice}
Recent work extended \piDILL and \cp with operators for non-deterministic behaviour~\cite{atkey2016,caires2014,caires2017}. These extensions all implement an operator known as non-deterministic local choice. (This operator is written as \tm{P+Q}, but should not be confused with input-guarded choice from the \textpi-calculus~\cite{milner1992b}.) Non-deterministic local choice can be summarised by the following typing and reduction rules:
\begin{center}
  \begin{prooftree*}
    \AXC{$\seq[{ P }]{ \Gamma }$}
    \AXC{$\seq[{ Q }]{ \Gamma }$}
    \BIC{$\seq[{ P + Q }]{ \Gamma }$}
  \end{prooftree*}
  \hspace*{2cm}
  \(
  \begin{array}{c}
    \reducesto{P + Q}{P}\\
    \reducesto{P + Q}{Q}
  \end{array}
  \)
\end{center}
Local choice introduces non-determinism explicitly, by listing all possible choices. This is unlike the \textpi-calculus, where non-determinism arises due to multiple processes communicating on shared channels. We can easily implement local choice in the \textpi-calculus, using a nullary communication:
\begin{center}
  \(
  \begin{array}{c}
    \tm{( \piPar{\piPar{\piUSend{x}{}{\piHalt}}{\piRecv{x}{}{P}}}{\piRecv{x}{}{Q}} )}
    \\[1ex]
    \rotatebox[origin=c]{270}{$\Longrightarrow^{\star}$}
    \\[1ex]
    \tm{( \piPar{P}{\piRecv{x}{}{Q}} )}
    \quad
    \text{or}
    \quad
    \tm{( \piPar{\piRecv{x}{}{P}}{Q} )}
  \end{array}
  \)
\end{center}
In this implementation, the process \tm{\piUSend{x}{}{0}} will ``unlock'' either \tm{P} or \tm{Q}, leaving the other process deadlocked. Or we could use input-guarded choice:
\begin{center}
  \(
  \tm{( \piPar{\piUSend{x}{}{\piHalt}}{( \piRecv{x}{}{P} + \piRecv{x}{}{Q} )} )}
  \)
\end{center}
However, there are many non-deterministic processes in the \textpi-calculus that are awkward to encode using non-deterministic local choice. Let us recall our example:
\begin{center}
  \(
  \begin{array}{c}
    \tm{(\piPar{%
    \piUSend{x}{\sliceofcake}{\piUSend{x}{\nope}{\store}}
    }{%
    \piPar{\piRecv{x}{y}{\ami}}{\piRecv{x}{z}{\boe}}
    })}
    \\[1ex]
    \rotatebox[origin=c]{270}{$\Longrightarrow^{\star}$}
    \\[1ex]
    \tm{(\piPar{\store}{\piPar{\piSub{\sliceofcake}{y}{\ami}}{
    \piSub{\nope}{z}{\boe}}})}
    \quad
    \text{or}
    \quad
    \tm{(\piPar{\store}{\piPar{\piSub{\nope}{y}{\ami}}{\piSub{
    \sliceofcake}{z}{\boe}}})}
  \end{array}
  \)
\end{center}
This non-deterministic interaction involves communication. If we wanted to write down a process which exhibited the same behaviour using non-deterministic local choice, we would have to write the following process:
\begin{center}
  \(
  \begin{array}{c}
    \tm{
    (\piPar{%
    \piUSend{x}{\sliceofcake}{\piUSend{y}{\nope}{\store}}
    }{%
    \piPar{\piRecv{x}{z}{\ami}}{\piRecv{y}{w}{\boe}}
    })
    +
    (\piPar{%
    \piUSend{y}{\sliceofcake}{\piUSend{x}{\nope}{\store}}
    }{%
    \piPar{\piRecv{x}{z}{\ami}}{\piRecv{y}{w}{\boe}}
    })
    }
    \\[1ex]
    \rotatebox[origin=c]{270}{$\Longrightarrow^{\star}$}
    \\[1ex]
    \tm{(\piPar{\store}{\piPar{\piSub{\sliceofcake}{y}{\ami}}{\piSub{\nope}{z}{\boe}}})}
    \quad
    \text{or}
    \quad
    \tm{(\piPar{\store}{\piPar{\piSub{\nope}{y}{\ami}}{\piSub{\sliceofcake}{z}{\boe}}})}
  \end{array}
  \)
\end{center}
In essence, instead of modelling a non-deterministic interaction, we are enumerating the resulting deterministic interactions. This means non-deterministic local choice cannot model non-determinism in the way the \textpi-calculus does.
Enumerating all possible outcomes becomes worse the more processes are involved in an interaction. Imagine the following scenario:
\begin{quote}
  Three customers, \Ami, \Boe, and \Cat, have a craving for cake. Should cake be sold out, however, well... a doughnut will do. They prepare to order their goods via an online store. Unfortunately, they all decide to use the same \emph{shockingly} under-stocked store, which has only one slice of cake, and a single doughnut. After that, all it can deliver is disappointment.
\end{quote}
We can model this scenario in the \textpi-calculus, where \ami, \boe, \cat, and \store are four processes modelling \Ami, \Boe, \Cat, and the store, and \sliceofcake, \doughnut, and \nope are three channels giving access to a slice of cake, a so-so doughnut, and disappointment, respectively.
\begin{center}
  \makebox[\textwidth][c]{\ensuremath{
    \begin{array}{c}
      \tm{(\piPar{%
      \piUSend{x}{\sliceofcake}{\piUSend{x}{\doughnut}{\piUSend{x}{\nope}{\store}}}
      }{%
      \piPar{\piRecv{x}{y}{\ami}}{\piPar{\piRecv{x}{z}{\boe}}{\piRecv{x}{w}{\cat}}
      })}}
      \\[1ex]
      \rotatebox[origin=c]{270}{$\Longrightarrow^{\star}$}
      \\[1ex]
      \tm{(\piPar{\store}{\piPar{\piSub{\sliceofcake}{y}{\ami}}{\piPar{\piSub{\doughnut}{z}{\boe}}{\piSub{\nope}{w}{\cat}}}})}
      \;\text{or}\;
      \tm{(\piPar{\store}{\piPar{\piSub{\sliceofcake}{y}{\ami}}{\piPar{\piSub{\nope}{z}{\boe}}{\piSub{\doughnut}{w}{\cat}}}})}
      \\[1ex]
      \tm{(\piPar{\store}{\piPar{\piSub{\doughnut}{y}{\ami}}{\piPar{\piSub{\nope}{z}{\boe}}{\piSub{\sliceofcake}{w}{\cat}}}})}
      \;\text{or}\;
      \tm{(\piPar{\store}{\piPar{\piSub{\doughnut}{y}{\ami}}{\piPar{\piSub{\sliceofcake}{z}{\boe}}{\piSub{\nope}{w}{\cat}}}})}
      \\[1ex]
      \tm{(\piPar{\store}{\piPar{\piSub{\nope}{y}{\ami}}{\piPar{\piSub{\sliceofcake}{z}{\boe}}{\piSub{\doughnut}{w}{\cat}}}})}
      \;\text{or}\;
      \tm{(\piPar{\store}{\piPar{\piSub{\nope}{y}{\ami}}{\piPar{\piSub{\doughnut}{z}{\boe}}{\piSub{\sliceofcake}{w}{\cat}}}})}
      \\[1ex]
    \end{array}
  }}
\end{center}
With the addition of one process, modelling \Cat, we have increased the number of possible outcomes enormously! In general, the number of outcomes for these types of scenarios is $n!$, where $n$ is the number of processes. This means that if we wish to translate any non-deterministic process to one using non-deterministic local choice, we can expect a factorial growth in the size of the term.

\section{Hypersequent Classical Processes}\label{sec:cp-revisited}
In this section, we introduce \hcp~\cite{kokke2018tlla}, the basis for our calculus \nodcap. 
The term language for \hcp is a variant of the \textpi-calculus~\cite{milner1992b}. In HCP, processes ($\tm{P}$, $\tm{Q}$, $\tm{R}$) communicate using names ($\tm{x}$, $\tm{y}$, $\tm{z}$, \dots). Each name is one of the two endpoints of a bidirectional communication channel~\cite{vasconcelos2012}. A channel is formed by connecting two endpoints using name restriction. This is in contrast to~\cref{sec:introduction,sec:local-choice}, where we used names to represent channels.
\begin{defi}[Terms]\label{def:hcp-terms}
  \[
    \begin{array}{lrll}
      \tm{P}, \tm{Q}, \tm{R}
        & ::=& \tm{\cpLink{x}{y}}    &\text{link}
      \\&\mid& \tm{\piHalt}          &\text{terminated process}
      \\&\mid& \tm{\piNew{x}{x'}{P}}  &\text{name restriction, ``cut''}
      \\&\mid& \tm{( \piPar{P}{Q} )} &\text{parallel composition, ``mix''}
      \\&\mid& \tm{\piSend{x}{y}{P}} &\text{output}
      \\&\mid& \tm{\piRecv{x}{y}{P}} &\text{input}
      \\&\mid& \tm{\piSend{x}{}{P}}  &\text{halt}
      \\&\mid& \tm{\cpWait{x}{}{P}}  &\text{wait}
      \\&\mid& \tm{\cpInl{x}{P}}     &\text{select left choice}
      \\&\mid& \tm{\cpInr{x}{P}}     &\text{select right choice}
      \\&\mid& \tm{\cpCase{x}{P}{Q}} &\text{offer binary choice}
      \\&\mid& \tm{\cpAbsurd{x}}     &\text{offer nullary choice}
    \end{array}
  \]
\end{defi}\noindent
The variables $\tm{x}$, $\tm{y}$, $\tm{z}$, $\tm{u}$, $\tm{v}$, and $\tm{w}$ range over channel endpoints. Occasionally, we use $\tm{a}$, $\tm{b}$, and $\tm{c}$ to range over \emph{free} endpoints, \ie those which are not connected to another endpoint. The construct $\tm{\cpLink{x}{y}}$ links two endpoints~\cite{sangiorgi1996,boreale1998}, forwarding messages received on \tm{x} to \tm{y} and vice versa. The construct $\tm{\piNew{x}{x'}{P}}$ creates a new channel by connecting endpoints $\tm{x}$ and $\tm{x'}$. By convention, we name dual endpoints using primes, \eg $\tm{x}$ and $\tm{x'}$. However, the primes are merely a naming convention, and do not denote co-names, \eg $\tm{x}$ and $\tm{x'}$ are not inherently dual, only under a $\nu$-binder $\tm{\piNew{x}{x'}{}}$. The construct $\tm{\piPar{P}{Q}}$ and composes two processes. In \tm{\piRecv{x}{y}{P}} and \tm{\piSend{x}{y}{P}}, round brackets denote input, square brackets denote output. We use bound output~\cite{sangiorgi1996}, meaning that both input and output bind a new name.

Terms in \hcp are identified up to structural congruence.
\begin{defi}[Structural congruence]\label{def:hcp-equiv}
  The structural congruence $\equiv$ is the congruence closure over terms which satisfies the following additional axioms:
  \[
    \setlength{\arraycolsep}{3pt}
    \begin{array}{llcllllcll}
        \hcpEquivLinkComm
      & \tm{\cpLink{x}{y}}
      & \equiv
      & \tm{\cpLink{y}{x}}
      \\ \hcpEquivMixComm
      & \tm{\piPar{P}{Q}}
      & \equiv
      & \tm{\piPar{Q}{P}}
      \\ \hcpEquivMixAss
      & \tm{\piPar{P}{( \piPar{Q}{R} )}}
      & \equiv
      & \tm{\piPar{( \piPar{P}{Q} )}{R}}
      \\ \hcpEquivMixHalt
      & \tm{\piPar{P}{\piHalt}}
      & \equiv
      & \tm{P}
      \\ \hcpEquivNewHalt
      & \tm{\piNew{x}{x'}{\piHalt}}
      & \equiv
      & \tm{\piHalt}
      \\ \hcpEquivNewComm
      & \tm{\piNew{x}{x'}{\piNew{y}{y'}{P}}}
      & \equiv
      & \tm{\piNew{y}{y'}{\piNew{x}{x'}{P}}}
      \\ \hcpEquivScopeExt
      & \tm{\piNew{x}{x'}{( \piPar{P}{Q} )}}
      & \equiv
      & \tm{\piPar{P}{\piNew{x}{x'}{Q}}}
      & \text{if }\tm{x},\tm{x'}\not\in\tm{P}
    \end{array}
  \]
\end{defi}

Channels in \hcp are typed using a session type system which is a conservative extension of linear logic.
\begin{defi}[Types]\label{def:cp-types}
  \[
    \begin{array}{lrlllrll}
      \ty{A}, \ty{B}, \ty{C}
        & ::=& \; \ty{A \tens B} &\text{independent channels}
      & &\mid& \; \ty{\one}      &\text{unit for} \; {\tens}
      \\&\mid& \; \ty{A \parr B} &\text{interdependent channels}
      & &\mid& \; \ty{\bot}      &\text{unit for} \; {\parr}
      \\&\mid& \; \ty{A \plus B} &\text{internal choice}
      & &\mid& \; \ty{\nil}      &\text{unit for} \; {\plus}
      \\&\mid& \; \ty{A \with B} &\text{external choice}
      & &\mid& \; \ty{\top}      &\text{unit for} \; {\with}
    \end{array}
  \]  
\end{defi}

Duality plays a crucial role in both linear logic and session types. In \hcp, the two endpoints of a channel are assigned dual types. This ensures that, for instance, whenever a process \emph{sends} across a channel, the process on the other end of that channel is waiting to \emph{receive}. Each type \ty{A} has a dual, written \ty{A^\bot}. Duality (\ty{\cdot^\bot}) is an involutive function on types.
\begin{defi}[Duality]\label{def:cp-negation}
  \[
    \setlength{\arraycolsep}{3pt}
    \begin{array}{lclclcllclclcl}
              \ty{(A \tens B)^\bot} &=& \ty{A^\bot \parr B^\bot}
      &\quad \ty{\one^\bot}        &=& \ty{\bot}
      &\quad \ty{(A \parr B)^\bot} &=& \ty{A^\bot \tens B^\bot}
      &\quad \ty{\bot^\bot}        &=& \ty{\one}
      \\      \ty{(A \plus B)^\bot} &=& \ty{A^\bot \with B^\bot}
      &\quad \ty{\nil^\bot}        &=& \ty{\top}
      &\quad \ty{(A \with B)^\bot} &=& \ty{A^\bot \plus B^\bot}
      &\quad \ty{\top^\bot}        &=& \ty{\nil}
    \end{array}
  \]
\end{defi}

Environments associate channels with types. Names in environments must be unique, and environments \ty{\Gamma} and \ty{\Delta} can only be combined ($\ty{\Gamma}, \ty{\Delta}$) if $\text{cn}(\ty{\Gamma}) \cap \text{cn}(\ty{\Delta}) = \varnothing$, where $\text{cn}{(\ty{\Gamma})}$ denotes the set of channel names in $\ty{\Gamma}$. 
\begin{defi}[Environments]\label{def:cp-environments}
  $\ty{\Gamma}, \ty{\Delta}, \ty{\Theta} ::= \tmty{x_1}{A_1}\dots\tmty{x_n}{A_n}$
\end{defi}

\hcp registers parallelism using hyper-environments. A hyper-environment is a multiset of environments. While names within environments must be unique, names may be shared between multiple environments in a hyper-environment. We write $\ty{\mathcal{G} \hsep \mathcal{H}}$ to combine two hyper-environments.
\begin{defi}[Hyper-environments]\label{def:hcp-hyper-environment}
  $\ty{\mathcal{G}}, \ty{\mathcal{H}} ::= \; \ty{\emptyhypercontext} \; \mid \; \ty{\mathcal{G} \hsep \Gamma}$
\end{defi}

Typing judgements associate processes with collections of typed channels.
\begin{defi}[Typing judgements]\label{def:cp-typing-judgement}
  A typing judgement $\seq[P]{\Gamma_1 \hsep \dots \hsep \Gamma_n}$ denotes that the process $\tm{P}$ consists of $n$ independent, but potentially entangled processes, each of which communicates according to its own protocol $\Gamma_i$. Typing judgements can be constructed using the inference rules below. 
  \\[0.5\baselineskip]
  {Structural rules}
  \begin{center}
    \hcpInfAx
    \hcpInfCut
  \end{center}
  \begin{center}
    \hcpInfMix
    \hcpInfHalt
  \end{center}
  {Logical rules}
  \begin{center}
    \hcpInfBoundTens
    \hcpInfParr
  \end{center}
  \begin{center}
    \hcpInfOne
    \hcpInfBot
  \end{center}
  \begin{center}
    \hcpInfPlus1
    \hcpInfPlus2
  \end{center}
  \begin{center}
    \hcpInfWith
  \end{center}
  \begin{center}
    \hcpInfNil
    \hcpInfTop
  \end{center}
\end{defi}

\paragraph{Alternative syntax.}
In $(\one)$, the only well-typed continuation $\tm{P}$ is the terminated process $\tm{\piHalt}$. We could use an alternative formulation of the rule, which combines $(\one)$ and $\textsc{H-Mix}_0$. However, as $\textsc{H-Mix}_0$ is used on its own, and not just in combination with $(\one)$, we chose the present formulation to avoid having multiple different representations of the terminated process in the language.

Reductions relate processes with their reduced forms.
\begin{defi}[Reduction]\label{def:hcp-reduction}
  Reductions are described by the smallest relation $\Longrightarrow$ on process
  terms closed under the rules below:
  \begin{gather*}
    \begin{array}{llcll}
      \hcpRedAxCut1
      & \tm{\piNew{x}{x'}{(\piPar{\cpLink{w}{x}}{P})}}
      & \Longrightarrow
      & \tm{\cpSub{w}{x'}{P}}
      \\
      \hcpRedBetaTensParr
      & \tm{\piNew{x}{x'}{(\piPar{\piSend{x}{y}{P}}{\piRecv{x'}{y'}{R}})}}
      & \Longrightarrow
      & \tm{\piNew{x}{x'}{\piNew{y}{y'}{(\piPar{P}{R})}}}
      \\
      \hcpRedBetaOneBot
      & \tm{\piNew{x}{x'}{(\piPar{\piSend{x}{}{P}}{\piRecv{x'}{}{Q}})}}
      & \Longrightarrow
      & \tm{\piPar{P}{Q}}
      \\
      \hcpRedBetaPlusWith1
      & \tm{\piNew{x}{x'}{(\piPar{\cpInl{x}{P}}{\cpCase{x'}{Q}{R}})}}
      & \Longrightarrow
      & \tm{\piNew{x}{x'}{(\piPar{P}{Q})}}
      \\
      \hcpRedBetaPlusWith2
      & \tm{\piNew{x}{x'}{(\piPar{\cpInr{x}{P}}{\cpCase{x'}{Q}{R}})}}
      & \Longrightarrow
      & \tm{\piNew{x}{x'}{(\piPar{P}{R})}}
    \end{array}
  \end{gather*}
  \vspace*{0.5\baselineskip}
  \begin{center}
    \begin{prooftree*}
      \AXC{$\reducesto{P}{P^\prime}$}
      \SYM{\hcpRedGammaNew}
      \UIC{$\reducesto{\piNew{x}{x'}{P}}{\piNew{x}{x'}{P^\prime}}$}
    \end{prooftree*}
    \begin{prooftree*}
      \AXC{$\reducesto{P}{P^\prime}$}
      \SYM{\hcpRedGammaMix}
      \UIC{$\reducesto{\piPar{P}{Q}}{\piPar{P^\prime}{Q}}$}
    \end{prooftree*}
  \end{center}
  \begin{prooftree}
    \AXC{$\tm{P}\equiv\tm{Q}$}
    \AXC{$\reducesto{Q}{Q^\prime}$}
    \AXC{$\tm{Q^\prime}\equiv\tm{P^\prime}$}
    \SYM{\hcpRedGammaEquiv}
    \TIC{$\reducesto{P}{P^\prime}$}
  \end{prooftree}
\end{defi}

We define unbound output in terms of bound output and link~\cite{lindley2015semantics}:
\[
  \begin{array}{c}
    \tm{\piUSend{x}{y}{P}} \triangleq \tm{\cpSend{x}{z}{\cpLink{y}{z}}{P}}
    \\
    \\
    \hcpInfUnboundTens
    \qquad
    \begin{array}{l}
      \tm{\piNew{x}{x'}{(\piPar{\piUSend{x}{y}{P}}{\piRecv{x'}{y'}{Q}})}}
      \\
      \qquad\qquad\Longrightarrow
      \tm{\piNew{x}{x'}{(\piPar{P}{\cpSub{y}{y'}{Q}})}}
    \end{array}
  \end{array}
\]

\subsection{Example}
\label{sec:hcp-example}
\hcp uses hyper-sequents to structure communication, and it is this structure which rules out deadlocked interactions. Let us go back to our example of a deadlocked interaction from \cref{sec:introduction}. If we want to type this interaction in \hcp, we run into a problem: to connect $\tm{x}$ and $\tm{y}$, and $\tm{z}$ and $\tm{w}$, such that we get a deadlock, we need to construct the following term:
\begin{center}
  \(
  \tm{\piNew{x}{x'}{\piNew{y}{y'}{(\piPar
        {\piRecv{x}{u}{\piUSend{y}{\sliceofcake}{\ami}}}
        {\piRecv{y'}{v}{\piUSend{x'}{\bill}{\boe}}})}}}.
  \)
\end{center}
However, there is no typing derivation for this term. We can construct a typing derivation down to the sequent below, but we cannot introduce \emph{both} name restrictions: the \textsc{Cut} rule eliminates a hypersequent separator, which ensures that it only ever connects two independent processes, but the sequent below only has \emph{one}.
\[
  \tm{
    \piRecv{x}{z}{\piUSend{y}{\sliceofcake}{\ami}}
    \ppar
    \piRecv{y'}{w}{\piUSend{x'}{\bill}{\boe}}
  }
  \vdash
  \setlength{\arraycolsep}{0pt}
  \begin{array}{lclcllclcllcl}
       \tm{x}            &\; : &\; \ty{\money^\bot} & \ty{\parr} & \ty{\bot},
    &\;\tm{y}            &\; : &\; \ty{\cake} & \ty{\tens} & \ty{\one},
    &\;\tm{\sliceofcake} &\; : &\; \ty{\cake^\bot} \hsep
    \\ \tm{x'}           &\; : &\; \ty{\money} & \ty{\tens} & \ty{\one},
    &\;\tm{y'}           &\; : &\; \ty{\cake^\bot} & \ty{\parr} & \ty{\bot},
    &\;\tm{\bill}        &\; : &\; \ty{\money^\bot}
  \end{array}
\]

\subsection{Metatheory}
\label{sec:hcp-metatheory}
\hcp enjoys subject reduction, termination, and progress~\cite{kokke2018tlla}.
\begin{lem}[Preservation for $\equiv$]\label{lem:hcp-preservation-equiv}
  If $\tm{P}\equiv\tm{Q}$, then $\seq[P]{\mathcal{G}}$ iff $\seq[Q]{\mathcal{G}}$.
\end{lem} 
\begin{proof}
  By induction on the derivation of $\tm{P}\equiv\tm{Q}$.
\end{proof}
\begin{thm}[Preservation]\label{thm:hcp-preservation}
  If $\seq[P]{\mathcal{G}}$ and $\reducesto{P}{Q}$, then $\seq[Q]{\mathcal{G}}$.
\end{thm} 
\begin{proof}
  By induction on the derivation of $\reducesto{P}{Q}$.
\end{proof}
\begin{defi}[Actions]
  A process $\tm{P}$ acts on $\tm{x}$ whenever $\tm{x}$ is free in the outermost
  term constructor of $\tm{P}$, \eg, $\tm{\piSend{x}{y}{P}}$ acts on $\tm{x}$
  but not on $\tm{y}$, and $\tm{\cpLink{x}{y}}$ acts on both $\tm{x}$ and $\tm{y}$.
  A process $\tm{P}$ is an action if it acts on some channel $\tm{x}$.
\end{defi}
\begin{defi}[Canonical forms]\label{def:hcp-canonical-forms}
  A process $\tm{P}$ is in canonical form if
  \[
  \tm{P} \equiv \tm{\piNew{x_1}{x'_1}{\dots\piNew{x_n}{x'_n}{(P_1 \mid \dots \mid P_{n+m+1})}}},
  \]
  such that: no process $\tm{P_i}$ is a cut or a mix; no process $\tm{P_i}$ is a link acting on a bound channel $\tm{x_i}$ or $\tm{x'_i}$; and no two processes $\tm{P_i}$ and $\tm{P_j}$ are acting on dual endpoints $\tm{x_i}$ and $\tm{x'_i}$ of the same channel.
\end{defi}
\begin{lem}
  If a well-typed process $\tm{P}$ is in canonical form, then it is blocked on
  an external communication, \ie,
  $\tm{P}\equiv\tm{\piNew{x_1}{x'_1}{\dots\piNew{x_n}{x'_n}{(P_1\mid\dots\mid P_{n+m+1})}}}$
  such that at least one process $\tm{P_i}$ acts on a free name.
\end{lem}
\begin{proof}
  We have
  \(
  \tm{P} \equiv \tm{\piNew{x_1}{x'_1}{\dots\piNew{x_n}{x'_n}{(P_1 \ppar \dots \ppar P_{n+m+1})}}},
  \)
  such that no $\tm{P_i}$ is a cut or a link acting on a bound channel, and no two processes $\tm{P_i}$ and $\tm{P_j}$ are acting on the endpoints of the same channel. The prefix of cuts and mixes introduces $n$ channels. Each application of cut requires an application of mix, so the prefix introduces $n+m+1$ processes. Therefore, at least $m+1$ of the processes $\tm{P_i}$ must be acting on a free channel, i.e., blocked on an external communication.
\end{proof}
\begin{thm}[Progress]\label{thm:hcp-progress}
  If $\seq[P]{\mathcal{G}}$, then either $\tm{P}$ is in canonical form, or there exists a process $\tm{Q}$ such that $\tm{P}\Longrightarrow\tm{Q}$.
\end{thm} 
\begin{proof}
  We consider the maximum prefix of cuts and mixes of $\tm{P}$ such that
  \[
  \tm{P} \equiv \tm{\piNew{x_1}{x'_1}{\dots\piNew{x_n}{x'_n}{(P_1 \ppar \dots \ppar P_{n+m+1})}}},
  \]
  and no $\tm{P_i}$ is a cut. If any process $\tm{P_i}$ is a link, we reduce by $(\cpLink{}{})$. If any two processes $\tm{P_i}$ and $\tm{P_j}$ are acting on dual endpoints $\tm{x_i}$ and $\tm{x'_i}$ of the same channel, we rewrite by $\equiv$ and reduce by the appropriate $\beta$-rule. Otherwise, $\tm{P}$ is in canonical form.
\end{proof}
\begin{thm}[Termination]\label{thm:hcp-termination}
  If $\seq[P]{\mathcal{G}}$, then there are no infinite $\Longrightarrow$-reduction sequences.
\end{thm} 
\begin{proof}
  Every reduction reduces a single cut to zero, one or two cuts. However, each of these cuts is smaller, measured in the size of the cut formula. Furthermore, each instance of the structural congruence preserves the size of the cut. Therefore, there cannot be an infinite $\Longrightarrow$-reduction sequence.
\end{proof}

\subsection{Erratum for HCP}\label{sec:erratum}
 The typing rules for HCP presented here are more restrictive than those in earlier publications~\cite{kokke2018tlla}. Progress does not hold for the earlier version. For instance, the following process is stuck, yet typeable:
 \begin{prooftree}
   \AXC{}
   \UIC{$\seq[\piHalt]{\emptyhypercontext}$}
   \UIC{$\seq[\piSend{y}{}{\piHalt}]{\tmty{y}{\one}}$}
   \UIC{$\seq[\piSend{x}{}{\piSend{y}{}{\piHalt}}]{\tmty{x}{\one}\hsep\tmty{y}{\one}}$}

   \AXC{}
   \UIC{$\seq[\cpLink{x}{z}]{\tmty{x}{\bot},\tmty{z}{\one}}$}
   \UIC{$\seq[\piRecv{y}{}{\cpLink{x}{z}}]{\tmty{x}{\bot},\tmty{y}{\bot},\tmty{z}{\one}}$}

   \BIC{$\seq%
     [\piSend{x}{}{\piSend{y}{}{\piHalt}}\ppar\piRecv{y}{}{\cpLink{x}{z}}]
     {\tmty{x}{\one}\hsep\tmty{y}{\one}\hsep\tmty{x}{\bot},\tmty{y}{\bot},\tmty{z}{\one}}$}

   \UIC{$\seq%
     [(\nu{y})({\piSend{x}{}{\piSend{y}{}{\piHalt}}\ppar\piRecv{y}{}{\cpLink{x}{z}}})]
     {\tmty{x}{\one}\hsep\tmty{x}{\bot},\tmty{z}{\one}}$}

   \UIC{$\seq%
     [(\nu{x})(\nu{y})({\piSend{x}{}{\piSend{y}{}{\piHalt}}\ppar\piRecv{y}{}{\cpLink{x}{z}}})]
     {\tmty{z}{\one}}$}
 \end{prooftree}
 The earlier typing rules failed to guarantee a crucial property: each typing environment should correspond to one top-level action. The rules presented in this paper fixes the problem by disallowing hyper-environments in logical rules.

 The move from channel names to endpoint names is not essential to the fix, but significantly streamlines the presentation. Otherwise, the type system must guarantee that each channel name occurs at most twice in the hypersequent, and if twice, then with dual types. Using endpoint names, it is sufficient to require that all names be distinct.

\section{Shared Channels and Non-determinism}\label{sec:cpnd}
In this section, we will discuss our main contribution: an extension of \hcp which allows for races while still excluding deadlocks. We have seen in \cref{sec:hcp-example} how \hcp excludes deadlocks, but how exactly does \hcp exclude races? Let us return to our example in \textpi-calculus from \cref{sec:introduction}, to the interaction between \Ami, \Boe and the store. 
\begin{center}
  \(
  \begin{array}{c}
    \tm{(\piPar{%
    \piUSend{x}{\sliceofcake}{\piUSend{x}{\nope}{\store}}
    }{%
    \piPar{\piRecv{x}{y}{\ami}}{\piRecv{x}{z}{\boe}}
    })}
    \\[1ex]
    \rotatebox[origin=c]{270}{$\Longrightarrow^{\star}$}
    \\[1ex]
    \tm{(\piPar{\store}{\piPar{\piSub{\sliceofcake}{y}{\ami}}{\piSub{\nope}{z}{\boe}}})}
    \quad
    \text{or}
    \quad
    \tm{(\piPar{\store}{\piPar{\piSub{\nope}{y}{\ami}}{\piSub{\sliceofcake}{z}{\boe}}})}
  \end{array}
  \)
\end{center}
Races occur when more than two processes attempt to communicate simultaneously over the \emph{same} channel. However, the \textsc{Cut} rule of \hcp requires that \emph{exactly two} processes communicate over each channel:
\begin{center}
  \hcpInfCut
\end{center}
We could attempt to write down a protocol for our example, stating that the store
has a pair of channels $\tm{x}, \tm{y} : \ty{\cake}$ with which it communicates
with \Ami and \Boe, taking \cake to be the type of interactions in which cake
\emph{may} be obtained, i.e.\ of both \sliceofcake and \nope, and state that the
store communicates with \Ami \emph{and} \Boe over a channel of type \ty{\cake
  \parr \cake}.
However, this \emph{only} models interactions such as the following:
\begin{prooftree}
  \AXC{$\seq[{ \ami }]{ \Gamma, \tmty{y}{\cake^\bot} }$}
  \AXC{$\seq[{ \boe }]{ \Delta, \tmty{x}{\cake^\bot} }$}
  \NOM{H-Mix}
  \BIC{$\seq[{ (\piPar{\ami}{\boe}) }]{
      \Gamma, \tmty{y}{\cake^\bot} \hsep \Delta, \tmty{x}{\cake^\bot} }$}
  \SYM{(\tens)}
  \UIC{$\seq[{ \cpSend{x}{y}{\ami}{\boe} }]{
      \Gamma, \Delta, \tmty{x}{\cake^\bot \tens \cake^\bot} }$}
  \AXC{$\seq[{ \store }]{ \Theta, \tmty{y'}{\cake}, \tmty{x'}{\cake} }$}
  \SYM{(\parr)}
  \UIC{$\seq[{ \cpRecv{x'}{y'}{\store} }]{
      \Theta, \tmty{x'}{\cake \parr \cake} }$}
  \NOM{H-Mix}
  \BIC{$\seq[{ (\piPar{\cpSend{x}{y}{\ami}{\boe}}{\cpRecv{x'}{y'}{\store}}) }]{
      \Gamma, \Delta, \tmty{x}{\cake^\bot \tens \cake^\bot}
      \hsep
      \Theta, \tmty{x'}{\cake \parr \cake} }$}
  \NOM{Cut}
  \UIC{$\seq[{ \piNew{x}{x'}{(\piPar{\cpSend{x}{y}{\ami}{\boe}}{\cpRecv{x'}{y'}{\store}})} }]{
      \Gamma, \Delta, \Theta }$}
\end{prooftree}
In this interaction, \Ami will get whatever the store decides to send on \tm{x}, and \Boe will get whatever the store decides to send on \tm{y}. This means that this interactions gives the choice of who receives what \emph{to the store}. This is not an accurate model of our original example, where the choice of who receives the cake is non-deterministic and depends on factors outside of any of the participants' control!

Modelling racy behaviour, such as that in our example, is essential to describing the interactions that take place in realistic concurrent systems. We would like to extend \hcp to allow such races in a way which mirrors the way in which the \textpi-calculus handles non-determinism. Let us return to our example:
\begin{center}
  \(
  \tm{(\piPar{%
      \piUSend{x}{\sliceofcake}{\piUSend{x}{\nope}{\store}}
    }{%
      \piPar{\piRecv{x}{y}{\ami}}{\piRecv{x}{z}{\boe}}
    })}
  \)
\end{center}
In this interaction, we see that the channel \tm{x} is only used as a way to connect the various clients, \Ami and \Boe, to the store. The \emph{real} communication, sending the slice of cake and disappointment, takes places on the channels \tm{\sliceofcake}, \tm{\nope}, \tm{y} and \tm{z}. Inspired by this, we add two new constructs to the term language of \hcp for sending and receiving on a \emph{shared} channel. These actions are marked with a \tm{\star} to distinguish them from ordinary sending and receiving. 
\begin{defi}[Terms]\label{def:nc-terms}
  We extend \cref{def:hcp-terms} as follows:
  \[\!
    \begin{aligned}
      \tm{P}, \tm{Q}, \tm{R}
          ::=& \; \dots
      \\ \mid& \; \tm{\ncCnt{x}{y}{P}} &&\text{client creation}
      \\ \mid& \; \tm{\ncSrv{x}{y}{P}} &&\text{server interaction}
    \end{aligned}
  \]
\end{defi}
As before, round brackets denote input, square brackets denote output. Note that $\tm{\ncCnt{x}{y}{P}}$, much like $\tm{\piSend{x}{y}{P}}$, is a bound output: both client creation and server interaction bind a new name.
The structural congruence, which identifies certain terms, is the same as \cref{def:hcp-equiv}.

In any non-deadlock interaction between a server and some clients, 
there must be \emph{exactly} as many clients as there are server interactions.
Therefore, we add two new \emph{dual} types for client pools and servers, which
track how many clients or server interactions they represent.
\begin{defi}[Types]\label{def:nc-types}
  We extend \cref{def:cp-types} as follows:
  \[
    \begin{array}{lrll}
      \ty{A}, \ty{B}, \ty{C}
        & ::=& \; \dots
      \\&\mid& \; \ty{\take[n]{A}} &\text{pool of} \; n \; \text{clients}
      \\&\mid& \; \ty{\give[n]{A}} &n \; \text{server interactions}
    \end{array}
  \]  
\end{defi}
The types $\ty{\take[n]{A}}$ and $\ty{\give[n]{A^{\bot}}}$ are dual. (The subscripts must be identical, \ie $\ty{\take[n]{A}}$ is not dual to $\ty{\give[m]{A^{\bot}}}$ if ${n}\neq{m}$.)
Duality remains an involutive function.

We have to add typing rules to associate our new client and server interactions
with their types. 
The definition for environments will remain unchanged, but we will extend the
definition for the typing judgement.
To determine the new typing rules, we essentially answer the question
``What typing constructs do we need to complete the following proof?''
\begin{prooftree}
  \AXC{$\seq[{ \ami }]{ \Gamma, \tmty{x'}{\cake^\bot} }$}
  \noLine\UIC{$\smash{\vdots}\vphantom{\vdash}$}
  \AXC{$\seq[{ \boe }]{ \Delta, \tmty{y'}{\cake^\bot} }$}
  \noLine\UIC{$\smash{\vdots}\vphantom{\vdash}$}
  \AXC{$\seq[{ \store }]{ \Theta, \tmty{z}{\cake}, \tmty{z'}{\cake} }$}
  \noLine\UIC{$\smash{\vdots}\vphantom{\vdash}$}
  \noLine\TIC{$\seq[{
      \piNew{x}{x'}{(\piPar{\ncPool{\ncCnt{x}{z}{\ami}}{\ncCnt{x}{z'}{\boe}}}{
        \ncSrv{x'}{w}{\ncSrv{x'}{w'}{\store}}})} }]{
      \Gamma, \Delta, \Theta }$}
\end{prooftree}
The constructs $\tm{\ncCnt{x}{y}{P}}$ and $\tm{\ncSrv{x}{y}{P}}$ introduce a single client or server action, respectively---hence, channels of type $\ty{\take[1]{}}$ and $\ty{\give[1]{}}$. However, when we cut, we want to cut on both interactions simultaneously. We need rules for the \emph{contraction} of shared channel names.

\subsection{Clients and Pooling}\label{sec:clients-and-pooling}
A client pool represents a number of independent processes, each wanting to interact with the same server. Examples of such a pool include \Ami and \Boe from our example, customers for online stores in general, and any number of processes which interact with a single, centralised server.

We introduce two new rules: one to construct clients, and one to pool them together. The first rule, $(\take[1]{})$, interacts over a channel as a client. It does this by receiving a channel $\tm{y}$ over a \emph{shared} channel $\tm{x}$. The channel $\tm{y}$ is the channel across which the actual interaction will eventually take place. The second rule, $\textsc{Cont}_{!}$, allows us to contract shared channel names with the same type. When used together with \textsc{H-Mix}, this allows us to pool clients together.
\begin{center}
  \ncInfTake1
  \ncInfPool
\end{center}%
Using these rules, we can derive the left-hand side of our proof by marking \Ami and \Boe as clients, and pooling them together.
\begin{prooftree}
  \AXC{$\seq[{ \ami }]{ \Gamma, \tmty{z}{\cake^\bot} }$}
  \SYM{(\take[1]{})}
  \UIC{$\seq[{ \ncCnt{x}{z}{\ami} }]{ \Gamma, \tmty{z}{\take[1]{\cake^\bot}} }$}

  \AXC{$\seq[{ \boe }]{ \Delta, \tmty{z'}{\cake^\bot} }$}
  \SYM{(\take[1]{})}
  \UIC{$\seq[{ \ncCnt{x'}{z'}{\boe} }]{ \Delta, \tmty{x'}{\take[1]{\cake^\bot}} }$}

  \NOM{H-Mix}
  \BIC{$\seq[{ \ncPool{\ncCnt{x}{z}{\ami}}{\ncCnt{x'}{z'}{\boe}} }]{
      \Gamma, \tmty{x}{\take[1]{\cake^\bot}} \hsep
      \Delta, \tmty{x'}{\take[1]{\cake^\bot}} }$}

  \SYM{\textsc{Cont}_{!}}
  \UIC{$\seq[{ \ncPool{\ncCnt{x}{z}{\ami}}{\ncCnt{x}{z'}{\boe}} }]{
      \Gamma, \Delta, \tmty{x}{\take[2]{\cake^\bot}} }$}
\end{prooftree}

\subsection{Servers and Sequencing}\label{sec:servers-and-sequencing}
Dual to a pool of $n$ clients in parallel is a server with $n$ actions in sequence. Our interpretation of a server is a process which offers some number of interdependent interactions of the same type. Examples include the store from our example, which gives out slices of cake and disappointment, online stores in general, and any central server which interacts with some number of client processes.

We introduce two new rules to construct servers. The first rule, $(\give[1]{})$, marks a interaction over some channel as a server interaction. It does this by sending a channel $\tm{y}$ over a \emph{shared} channel $\tm{x}$. The channel $\tm{y}$ is the channel across which the actual interaction will take place. The second rule, $\textsc{Cont}_{?}$, allows us to merge two (possibly interleaved) sequences of server interactions. This allows us to construct a server which has multiple interactions of the same type, across the same shared channel.
\begin{center}
  \ncInfGive1
  \ncInfCont
\end{center}
Using these rules, we can derive the right-hand side of our proof, by marking each of the store's interactions as server interactions, and then contracting them.
\begin{prooftree}
  \AXC{$\seq[{ \store }]{ \Theta, \tmty{w}{\cake}, \tmty{w'}{\cake} }$}
  \SYM{(\give[1]{})}
  \UIC{$\seq[{ \ncSrv{y'}{w'}{\store} }]{
      \Theta, \tmty{z}{\cake}, \tmty{y'}{\give[1]{\cake}} }$}
  \SYM{(\give[1]{})}
  \UIC{$\seq[{ \ncSrv{y}{w}{\ncSrv{y'}{w'}{\store}} }]{
      \Theta, \tmty{y}{\give[1]{\cake}}, \tmty{y'}{\give[1]{\cake}} }$}
  \SYM{\textsc{Cont}_{?}}
  \UIC{$\seq[{ \ncSrv{y}{w}{\ncSrv{x}{w'}{\store}} }]{
      \Theta, \tmty{y}{\give[2]{\cake}} }$}
\end{prooftree}
Thus, we complete the typing derivation of our example.

\begin{defi}[Typing judgements]\label{def:nc-typing-judgement}
  We extend \cref{def:cp-typing-judgement} as follows:
  {\normalfont
    \begin{center} \ncInfTake1 \ncInfGive1 \end{center}
    \begin{center} \ncInfPool  \ncInfCont  \end{center}
  }
\end{defi}

\subsection{Running Clients and Servers}\label{sec:nc-running-clients-and-servers}
Finally, we need to extend the reduction rules to allow for the reduction of client and server processes. The reduction rule we add is a variant of $\hcpRedBetaTensParr$.
\begin{defi}[Reduction]\label{def:nc-reduction}
  We extend \cref{def:hcp-reduction} as follows:
  \[
    \begin{array}{llll}
      \ncRedBetaStar{}
      & \tm{\piNew{x}{x'}{(\piPar{(\piPar{\ncCnt{x}{y}{P}}{\ncSrv{x'}{y'}{Q}})}{R})}}
      & \Longrightarrow \;
      & \tm{\piNew{x}{x'}{(\piPar{\piNew{y}{y'}{(\piPar{P}{Q})}}{R})}}
    \end{array}
  \]
\end{defi}
The difference between $\ncRedBetaStar{}$ and $\hcpRedBetaTensParr$ is that the former allows reduction to happen in the presence of an unrelated process $\tm{R}$, which is passed along unchanged. This is necessary, as there may be other clients waiting to interact with the server on the shared channel $\tm{x}$, which cannot be moved out of scope of the name restriction $\piNew{x}{}$. When there is no unrelated process $\tm{R}$, \ie, when there is only a single client, we can rewrite by $\hcpEquivMixHalt$ before and after applying $\ncRedBetaStar{}$.

So where does the non-determinism in \nodcap come from? Let us say we have a term of the following form:
\[
  \tm{
    \piNew{x}{x'}{(\piPar
    {\ncPool{\ncCnt{x}{y_1}{P_1}}{\dots \mid \ncCnt{x}{y_n}{P_n}}}
    {\ncSrv{x'}{y'_1}{\dots\ncSrv{x'}{y'_n}{Q}}})}
  }
\]
As parallel composition is commutative and associative, we can rewrite this term to pair any client in the pool with the server before applying $\ncRedBetaStar{}$. Thus, like in the \textpi-calculus, the non-determinism is introduced by the structural congruence.

Does this mean that, for an arbitrary client pool $\tm{P}$ in $\tm{\piNew{x}{y}{(\piPar{P}{\ncSrv{y}{w}{Q}})}}$, every client in that pool is competing for the server interaction on $\tm{x}$? Not necessarily, as some portion of the clients can be blocked on an external communication. For instance, in the term below, clients $\tm{\ncCnt{x}{z_{n+1}}{P_{n+1}}} \dots \tm{\ncCnt{x}{z_m}{P_m}}$ are blocked on a communication on the external channel $\tm{a}$:
\[
  \arraycolsep=0pt
  \tm{
  \begin{array}{lrl}
    \piNew{x}{x'}{}
    &  ((&\; \ncPool{\ncCnt{x}{y_1}{P_1}}{\dots\mid\ncCnt{x}{y_n}{P_n}}\\
    &\mid&\; \cpWait{a}{\ncPool{\ncCnt{x}{y_{n+1}}{P_{n+1}}}{\dots \mid \ncCnt{x}{y_m}{P_m}}}\;)\\
    &\mid&\; \ncSrv{x'}{y'_1}{\dots\ncSrv{x'}{y'_m}{Q}}\;)
  \end{array}}
\]
If we reduce this term, then only the clients $\tm{\ncCnt{x}{z_1}{P_1}} \dots \tm{\ncCnt{x}{z_n}{P_n}}$ will be assigned server interactions, and we end up with the following canonical form:
\[
  \arraycolsep=0pt
  \tm{
  \begin{array}{lrl}
    \piNew{x}{x'}{}
    &   (&\; \cpWait{a}{\ncPool{\ncCnt{x}{y_{n+1}}{P_{n+1}}}{\dots\mid\ncCnt{x}{y_m}{P_m}}}\\
    &\mid&\; \ncSrv{x'}{y'_{n+1}}{\dots\ncSrv{x'}{y'_m}{Q}}\;)
  \end{array}}
\]
This matches our intuition and the behaviour of the \textpi-calculus. For instance, we can now encode our example, where \Ami and \Boe both send a request for cake to the store, and the store sends back either a cake or nothing:
\[
  \tm{\piNew{x}{x'}{}
    \left(
      \begin{array}{l}
        \ncCnt{x}{x}{\piRecv{x}{y}{\ami}}\ppar
        \\
        \ncCnt{x}{x}{\piRecv{x}{z}{\boe}}\ppar
        \\
        \ncSrv{x'}{x'_1}{\piUSend{x'_1}{\sliceofcake}{\ncSrv{x'}{x'_2}{\piUSend{x'_2}{\nope}{\store}}}}
      \end{array}
    \right)
  }
  \Longrightarrow^\star
  \begin{array}{c}
    \tm{(\cpSub{\sliceofcake}{y}{\ami}\ppar\cpSub{\nope}{z}{\boe}\ppar\store)}
    \\
    \text{or}
    \\
    \tm{(\cpSub{\nope}{y}{\ami}\ppar\cpSub{\sliceofcake}{z}{\boe}\ppar\store)}
  \end{array}
\]
The encoding presented above is slightly more complex than necessary: after the store receives a request as $\tm{x'_1}$, it could simply perform the cake interaction over that channel, and similarly for $\tm{x'_2}$. However, we include these actions for clarity.

\paragraph{Alternative syntax.}
If we choose to reuse the terms $\tm{\piSend{x}{y}{P}}$ and $\tm{\piRecv{x}{y}{P}}$ for shared channels, we could replace $\hcpRedBetaTensParr$ with $\ncRedBetaStar{}$, using the latter rule for both cases.

\subsection{Metatheory}
\label{sec:nc-metatheory}
\nodcap enjoys subject reduction, termination, and progress.
\begin{lem}[Preservation for $\equiv$]\label{lem:nc-preservation-equiv}
  If $\tm{P}\equiv\tm{Q}$ and $\seq[P]{\mathcal{G}}$, then $\seq[Q]{\mathcal{G}}$.
\end{lem} 
\begin{proof}
  By induction on the derivation of $\tm{P}\equiv\tm{Q}$.
\end{proof}
\begin{thm}[Preservation]\label{thm:nc-preservation}
  If $\seq[P]{\mathcal{G}}$ and $\reducesto{P}{Q}$, then $\seq[Q]{\mathcal{G}}$.
\end{thm} 
\begin{proof}
  By induction on the derivation of $\reducesto{P}{Q}$.
\end{proof}
\begin{defi}[Actions]
  A process $\tm{P}$ acts on $\tm{x}$ whenever $\tm{x}$ is free in the outermost term constructor of $\tm{P}$, \eg, $\tm{\ncSrv{x}{y}{P}}$ acts on $\tm{x}$ but not on $\tm{y}$, and $\tm{\cpLink{x}{y}}$ acts on both $\tm{x}$ and $\tm{y}$. A process $\tm{P}$ is an action if it acts on some channel $\tm{x}$.
\end{defi}
\begin{defi}[Canonical forms]\label{def:nc-canonical-forms}
  A process $\tm{P}$ is in canonical form if
  \[
  \tm{P} \equiv \tm{\piNew{x_1}{x'_1}{\dots\piNew{x_n}{x'_n}{(P_1 \mid \dots \mid P_{n+m+1})}}},
  \]
  such that: no process $\tm{P_i}$ is a cut or a mix; no process $\tm{P_i}$ is a link acting on a bound channel $\tm{x_i}$ or $\tm{x'_i}$; and no two processes $\tm{P_i}$ and $\tm{P_j}$ are acting on dual endpoints $\tm{x_i}$ and $\tm{x'_i}$ of the same channel.
\end{defi}
\begin{lem}
  If a well-typed process $\tm{P}$ is in canonical form, then it is blocked on
  an external communication, \ie,
  $\tm{P}\equiv\tm{\piNew{x_1}{x'_1}{\dots\piNew{x_n}{x'_n}{(P_1\mid\dots\mid P_{n+m+1})}}}$
  such that at least one process $\tm{P_i}$ acts on a free name.
\end{lem}
\begin{proof}
  We have
  \(
  \tm{P} \equiv \tm{\piNew{x_1}{x'_1}{\dots\piNew{x_n}{x'_n}{(P_1 \ppar \dots \ppar P_{n+m+1})}}},
  \)
  such that no $\tm{P_i}$ is a cut or a link acting on a bound channel, and no two processes $\tm{P_i}$ and $\tm{P_j}$ are acting on the same bound channel with dual actions. The prefix of cuts and mixes introduces $n$ channels. Each application of cut requires an application of mix, so the prefix introduces $n+m+1$ processes. Each application of $\textsc{Cont}_{!}$ requires an application of mix, so there are at most $m$ clients acting on the same bound channel. Therefore, at least \emph{one} of the processes $\tm{P_i}$ must be acting on a free channel, i.e., blocked on an external communication.
\end{proof}
\begin{thm}[Progress]\label{thm:nodcap-progress}
  If $\seq[P]{\mathcal{G}}$, then either $\tm{P}$ is in canonical form, or there exists a process $\tm{Q}$ such that $\tm{P}\Longrightarrow\tm{Q}$.
\end{thm} 
\begin{proof}
  We consider the maximum prefix of cuts and mixes of $\tm{P}$ such that
  \[
  \tm{P} \equiv \tm{\piNew{x_1}{x'_1}{\dots\piNew{x_n}{x'_n}{(P_1 \ppar \dots \ppar P_{n+m+1})}}},
  \]
  and no $\tm{P_i}$ is a cut. If any process $\tm{P_i}$ is a link, we reduce by $(\cpLink{}{})$. If any two processes $\tm{P_i}$ and $\tm{P_j}$ are acting dual endpoints $\tm{x_i}$ and $\tm{x'_i}$ of the same channel, we rewrite by $\equiv$ and reduce by the appropriate $\beta$-rule. Otherwise, $\tm{P}$ is in canonical form.
\end{proof}
\begin{thm}[Termination]\label{thm:nodcap-termination}
  If $\seq[P]{\mathcal{G}}$, then there are no infinite $\Longrightarrow$-reduction sequences.
\end{thm} 
\begin{proof}
  Every reduction reduces a single cut to zero, one or two cuts. However, each of these cuts is smaller, measured in the size of the cut formula. Furthermore, each instance of the structural congruence preserves the size of the cut. Therefore, there cannot be an infinite $\Longrightarrow$-reduction sequence.
\end{proof}

\subsection{\nodcap and Non-deterministic Local Choice}\label{sec:nc-local-choice}
In \cref{sec:local-choice}, we discussed the non-deterministic local choice operator, which is used in several extensions of \piDILL and \cp~\cite{atkey2016,caires2014,caires2017}. This operator is admissible in \nodcap. We can derive the non-deterministic choice \tm{P+Q} by constructing the following term:
\[%
  \arraycolsep=0pt
  \tm{
  \begin{array}{lrlrl}
    \piNew{x}{x'}{}
    &((  & \; \ncCnt{x}{y}{\cpInl{y}{\cpHalt{y}}} \\
    &\mid& \; \ncCnt{x}{z}{\cpInr{z}{\cpHalt{z}}} \; )\\
    &\mid& \; \ncSrv{x'}{y'}{\ncSrv{x'}{z'}{y'\:{\triangleright}}}\\
    &    & \quad
           \begin{array}{rl}
             \{ & \texttt{inl}: \; \piNew{w}{w'}{(\piPar{\cpCase{z'}{\cpWait{z'}{\cpHalt{w}}}{
                  \cpWait{z'}{\cpHalt{w}}}}{\cpWait{w'}{P}})}
             \\
             ; & \texttt{inr}: \; \piNew{w}{w'}{(\piPar{\cpCase{z'}{\cpWait{z'}{\cpHalt{w}}}{
                 \cpWait{z'}{\cpHalt{w}}}}{\cpWait{w'}{Q}})} \; \})
           \end{array}
  \end{array}
  }
\]
This term is a cut between two processes.
\begin{itemize}
\item
  On the left-hand side, we have a pool of two processes, $\tm{\ncCnt{x}{y}{\cpInl{y}{\cpHalt{y}}}}$ and $\tm{\ncCnt{x}{z}{\cpInr{z}{\cpHalt{z}}}}$. Each makes a choice: the first sends $\tm{\texttt{inl}}$, and the second sends $\tm{\texttt{inr}}$. 
\item
  On the right-hand side, we have a server with both $\tm{P}$ and $\tm{Q}$. This server has two channels on which a choice is offered, $\tm{y'}$ and $\tm{z'}$. The choice on $\tm{y'}$ selects between $\tm{P}$ and $\tm{Q}$. The choice on $\tm{z'}$ does not affect the outcome of the process at all. Instead, it is discarded.
\end{itemize}
When these clients and the server are put together, the choices offered by the server will be non-deterministically lined up with the clients which make choices, and either $\tm{P}$ or $\tm{Q}$ will run.

While there is a certain amount of overhead involved in this encoding, it scales linearly in terms of the number of processes. The reverse---encoding the non-determinism present in \nodcap using non-deterministic local choice---scales exponentially, see, \eg, the examples in \cref{sec:local-choice}.

\section{Cuts with Leftovers}\label{sec:leftovers}
So far, our account of a non-determinism in client/server interactions only allows for interactions between equal numbers of clients and server interactions. A natural question is whether or not we can deal with the scenario in which there are more client than server interactions or vice versa, \ie, whether or not the following rules are derivable:
\begin{center}
  \begin{prooftree*}
    \AXC{$\seq{ \Gamma, \ty{\take[n+m]{A}} }$}
    \AXC{$\seq{ \Delta, \ty{\give[n]{A^\bot}} }$}
    \BIC{$\seq{ \Gamma, \Delta, \ty{\take[m]{A}} }$}
  \end{prooftree*}
  \begin{prooftree*}
    \AXC{$\seq{ \Gamma, \ty{\take[n]{A}} }$}
    \AXC{$\seq{ \Delta, \ty{\give[n+m]{A^\bot}} }$}
    \BIC{$\seq{ \Gamma, \Delta, \ty{\give[m]{A^\bot}} }$}
  \end{prooftree*}
\end{center}
These rules are derivable using a link. For instance, we can derive the rule for the case in which there are more clients than servers as follows:
\begin{prooftree}
  \AXC{$\seq[{ P }]{ \Gamma, \tmty{x}{\take[n+m]{A}} }$}
  \AXC{$\seq[{ Q }]{ \Delta, \tmty{x'}{\give[n]{A^\bot}} }$}
  \AXC{$\seq[{ \cpLink{x''}{w} }]{
      \tmty{x''}{\give[m]{A^{\bot}}}, \tmty{w}{\take[m]{A}} }$}
  \NOM{H-Mix}
  \BIC{$\seq[{ (\piPar{Q}{\cpLink{x''}{w}}) }]{
      \Delta, \tmty{x'}{\give[n]{A^\bot}} \hsep
      \tmty{x''}{\give[m]{A^{\bot}}}, \tmty{w}{\take[m]{A}} }$}
  \SYM{\textsc{Cont}_{!}}
  \UIC{$\seq[{ {(\piPar{Q}{\cpLink{x'}{w}})} }]{
      \Delta, \tmty{x'}{\give[n+m]{A^\bot}}, \tmty{w}{\take[m]{A}} }$}
  \NOM{H-Mix}
  \BIC{$\seq[{ (\piPar{P}{(\piPar{Q}{\cpLink{x'}{w}})}) }]{
      \Gamma, \tmty{x}{\take[n+m]{A}} \hsep
      \Delta, \tmty{x'}{\give[n+m]{A^\bot}}, \tmty{w}{\take[m]{A}} }$}
  \NOM{Cut}
  \UIC{$\seq[{ \piNew{x}{x'}{(\piPar{P}{(\piPar{Q}{\cpLink{x'}{w}})})} }]{
      \Gamma, \Delta, \tmty{w}{\take[m]{A}} }$}
\end{prooftree}

\section{Relation to Manifest Sharing}\label{sec:manifest}
In \cref{sec:local-choice}, we mentioned related work which extends \piDILL and \cp with non-deterministic local choice~\cite{atkey2016,caires2014,caires2017}, and contrasted these approaches with ours.
In this section, we will contrast our work with the more recent work on manifest sharing~\cite{balzer2017}.

Manifest sharing extends the session-typed language \SILL with two connectives, $\ty{\acquire{A}}$ and $\ty{\release{A}}$, which represent the places in a protocol where a shared resource is aquired and released, respectively. In the resulting language, \SILLS, we can define a type for, \eg, shared queues (using the notation for types introduced in this paper):
\[
  \ty{\queue{A}} ::=
  \ty{\acquire{
      (\; A^\bot \parr \release{(\queue{A})} \;)
      \with
      (\; (A \plus \bot) \parr \release{(\queue{A})} \;)
    }}
\]
The type $\ty{\queue{A}}$ types a shared channel which, after we aqcuire exclusive access, gives us the choice between enqueuing a value ($\ty{A^\bot}$) and releasing the queue, or dequeuing a value if there is any ($\ty{A \plus \bot}$) and releasing the queue.

The language \SILLS is much more expressive than \nodcap, as it has support for both shared channels and recursion.
In fact, Balzer, Pfenning, and Toninho~\cite{balzer2018} show that \SILLS is expressive enough to embed the untyped asynchronous \textpi-calculus. This expressiveness comes with a cost, as \SILLS processes are not guaranteed to be deadlock free, though recent work addresses this issue~\cite{balzer2019}.

Despite the difference in expressiveness, there are some similarities between \nodcap and \SILLS. In the former, shared channels represent (length-indexed) streams of interactions of the same type. In the latter, it is necessary for type preservation that shared channels are always released at the same type at which they were acquired, meaning that shared channels also represent (possibly infinite) streams of interactions of the same type. In fact, in \nodcap, the type for queues (with $n$ interactions) can be written as $\ty{\take[n]{(A^\bot \with (A \plus \bot))}}$.

One key difference between \nodcap and \SILLS is that in \SILLS a server must finish interacting with one client before interacting with another, whereas in \nodcap the server may interact with multiple clients simultaneously.

\section{Discussion and Future Work}\label{sec:conclusion}
We presented \nodcap, an extension of \hcp which permits non-deterministic communication without losing the strong connection to logic. We gave proofs for preservation, progress, and termination for the term reduction system of \nodcap. We showed that we can define non-deterministic local choice in \nodcap.

Our formalism so far has only captured servers that provide for a fixed number of clients.  More realistically, we would want to define servers that provide for arbitrary numbers of clients.  This poses two problems: how would we define arbitrarily-interacting stateful processes, and how would we extend the typing discipline of \nodcap to account for them without losing its static guarantees.

One approach to defining server processes would be to combine \nodcap with structural recursion and corecursion, following the $\mu\text{CP}$ extension of Lindley and Morris~\cite{lindley2016}.  Their approach can express processes which produce streams of \ty{A} channels. Such a process would expose a channel with the co-recursive type \ty{\nu X. A \parr (1 \plus X)}.  Given such a process, it is possible to produce a channel of type \ty{A \parr A \parr \cdots \parr A} for any number of \ty{A}s, allowing us to satisfy the type \ty{\give[n]{A}} for an arbitrary $n$.

We would also need to extend the typing discipline to capture arbitrary use of shared channels.  One approach would be to introduce resource variables and quantification.  Following this approach, in addition to having types \ty{\give[n] A} and \ty{\take[n] A} for concrete $n$, we would also have types \ty{\give[x] A} and \ty{\take[x] A} for resource variables $x$.  These variables would be introduced by quantifiers \ty{\forall x A} and \ty{\exists x A}. Defining terms corresponding to \ty{\forall x A}, and its relationship with structured recursion, presents an interesting area of further work.

Our account of \hcp did not include the exponentials \ty{\give A} and \ty{\take A}. The type \ty{\take A} denotes arbitrarily many independent instances of \ty{A}, while the type \ty{\give A} denotes a concrete (if unspecified) number of potentially-dependent instances of \ty{A}.  Existing interpretations of linear logic as session types have taken \ty{\take A} to denote \ty{A}-servers, while \ty{\give A} denotes \ty{A}-clients.  However, the analogy is imperfect: while we expect servers to provide arbitrarily many instances of their behaviour, we also expect those instances to be interdependent.

With quantification over resource variables, we can give precise accounts of both \cp's exponentials and idealised servers and clients. \cp exponentials could be embedded into this framework using the definitions $\ty{\take{A}} ::= \ty{\forall{n}\take[n]{A}}$ and $\ty{\give{A}} ::= \ty{\exists{n}{\give[n]{A}}}$. We would also have types that precisely matched our intuitions for server and client behavior: an \ty{A} server is of type \ty{\forall{n}{\give[n] A}}, as it serves an unbounded number of requests with the requests being interdependent, while a collection of \ty{A} clients is of type \ty{\exists{n}{\take[n] A}}, as we have a specific number of clients with each client being independent.

\bibliographystyle{alpha}
\bibliography{main}

\end{document}